\documentclass{amsart}

\usepackage{amssymb,amsmath,amscd,amsthm}

\date{\today}

\newcommand{\Z}{{\mathbb Z}}
\newcommand{\R}{{\mathbb R}}
\newcommand{\C}{{\mathbb C}}
\newcommand{\N}{{\mathbb N}}

\newcommand{\D}{{\mathbb D}}

\newcommand{\PP}{{\mathbb P}}
\newcommand{\E}{{\mathcal E}}

\newcommand{\supp}{{\mathrm{supp}}}
\newtheorem{theorem}{Theorem}[section]
\newtheorem{lemma}[theorem]{Lemma}

\newtheorem{coro}[theorem]{Corollary}

\theoremstyle{definition}
\newtheorem{remark}[theorem]{Remark}
\theoremstyle{definition}
\newtheorem{defi}[theorem]{Definition}

\renewcommand{\Im}{\mathrm{Im} \, }
\renewcommand{\Re}{\mathrm{Re} \, }

\newcommand{\abs}[1]{\left\vert#1\right\vert}
\newcommand{\set}[1]{\left\{#1\right\}}

\newcommand{\norm}[1]{\left\|#1\right\|}

\begin{document}

\title[Szeg\H{o} Cocycles, Random CMV Matrices, and the Ising Model]{Uniform Hyperbolicity for Szeg\H{o} Cocycles and Applications to Random CMV Matrices and the Ising Model}

\author[D.\ Damanik]{David Damanik}

\address{Department of Mathematics, Rice University, Houston, TX~77005, USA}

\email{damanik@rice.edu}

\thanks{D.\ D.\ was supported in part by NSF grant DMS--1067988.}

\author[J.\ Fillman]{Jake Fillman}

\address{Department of Mathematics, Rice University, Houston, TX~77005, USA}

\email{jdf3@rice.edu}

\thanks{J.\ F.\ was supported in part by NSF grant DMS--1067988.}

\author[M.\ Lukic]{Milivoje Lukic}

\address{Department of Mathematics, Rice University, Houston, TX~77005, USA}

\email{milivoje.lukic@rice.edu}

\thanks{M.\ L.\ was supported in part by NSF grant DMS--1301582.}

\author[W.\ Yessen]{William Yessen}

\address{Department of Mathematics, Rice University, Houston, TX~77005, USA}

\email{yessen@rice.edu}

\thanks{W.\ Y.\ was supported by the NSF Mathematical Sciences Postdoctoral Research Fellowship DMS-1304287}

\begin{abstract}
We consider products of the matrices associated with the Szeg\H{o} recursion from the theory of orthogonal polynomials on the unit circle and show that under suitable assumptions, their norms grow exponentially in the number of factors. In the language of dynamical systems, this result expresses a uniform hyperbolicity statement. We present two applications of this result. On the one hand, we identify explicitly the almost sure spectrum of extended CMV matrices with non-negative random Verblunsky coefficients. On the other hand, we show that no Ising model in one dimension exhibits a phase transition. Also, in the case of dynamically generated interaction couplings, we describe a gap labeling theorem for the Lee-Yang zeros in the thermodynamic limit.
\end{abstract}

\maketitle

\section{Introduction}\label{s.intro}

In this paper we are interested in two seemingly unrelated problems. Namely, we study the almost sure spectrum of extended CMV matrices with random Verblunsky coefficients and the zeros of the partition function of the one-dimensional Ising model with general interaction parameters. It turns out, however, that there is a common issue in both of these scenarios. Namely, it is essential to determine for which parameters a certain one-parameter family of cocycles is uniformly hyperbolic. In other words, by solving a dynamical systems problem we are able to address questions in spectral theory and statistical mechanics.

Let us be more specific and describe the questions we will study in this paper explicitly.

There is a well known one-to-one correspondence between probability measures on the unit circle and a class of five-diagonal matrices, the so-called CMV matrices; we refer the reader to \cite{S05, S05b} for background and further information. A CMV matrix is a semi-infinite matrix of the form
$$
\mathcal{C} = \begin{pmatrix}
{}& \bar\alpha_0 & \bar\alpha_1 \rho_0 & \rho_1
\rho_0
& 0 & 0 & \dots & {} \\
{}& \rho_0 & -\bar\alpha_1 \alpha_0 & -\rho_1
\alpha_0
& 0 & 0 & \dots & {} \\
{}& 0 & \bar\alpha_2 \rho_1 & -\bar\alpha_2 \alpha_1 &
\bar\alpha_3 \rho_2 & \rho_3 \rho_2 & \dots & {} \\
{}& 0 & \rho_2 \rho_1 & -\rho_2 \alpha_1 &
-\bar\alpha_3
\alpha_2 & -\rho_3 \alpha_2 & \dots & {} \\
{}& 0 & 0 & 0 & \bar\alpha_4 \rho_3 & -\bar\alpha_4
\alpha_3
& \dots & {} \\
{}& \dots & \dots & \dots & \dots & \dots & \dots & {}
\end{pmatrix}
$$
where $\alpha_n \in \D = \{ w \in \C : |w| < 1 \}$ and $\rho_n = (1-|\alpha_n|^2)^{1/2}$ -- it is easy to check that $\mathcal{C}$ defines a unitary operator on $\ell^2(\Z_+)$.

CMV matrices $\mathcal{C}$ are in one-to-one correspondence to probability measures $\mu$ on the unit circle $\partial \D$ that are not supported by a finite set. To go from $\mathcal{C}$ to $\mu$, one invokes the spectral theorem. To go from $\mu$ to $\mathcal{C}$, one can proceed either via orthogonal polynomials or via Schur functions. In the approach via orthogonal polynomials, the $\alpha_n$'s arise as recursion coefficients for the polynomials.

Explicitly, consider the Hilbert space $L^2(\partial \D,d\mu)$ and apply the Gram-Schmidt orthonormalization procedure to the
sequence of monomials $1 , z , z^2 , z^3 , \ldots$. This yields a sequence $\varphi_0 , \varphi_1 , \varphi_2 , \varphi_3 , \ldots$ of normalized polynomials that are pairwise orthogonal in $L^2(\partial \D,d\mu)$. Corresponding to $\varphi_n$, consider the ``reflected polynomial'' $\varphi_n^*$, where the coefficients of $\varphi_n$ are conjugated and then written in reverse order. Then, we have
\begin{equation}\label{e.tmbasic}
\begin{pmatrix} \varphi_{n+1}(z) \\ \varphi_{n+1}^*(z) \end{pmatrix} = \rho_n^{-1} \left( \begin{array}{cc} z & - \bar{\alpha}_n \\ - \alpha_n z& 1 \end{array} \right) \begin{pmatrix} \varphi_{n}(z) \\ \varphi_{n}^*(z) \end{pmatrix}
\end{equation}
for suitably chosen $\alpha_n \in \D$ (and again with $\rho_n = (1-|\alpha_n|^2)^{1/2}$). With these $\alpha_n$'s, one may form the corresponding CMV matrix $\mathcal{C}$ as above and obtain a unitary matrix for which the spectral measure corresponding to the cyclic vector $\delta_0$ is indeed the measure $\mu$ upon which we based the construction.

Depending on whether one starts out with the coefficients or the measure in this one-to-one correspondence, one obtains direct and inverse spectral theory in this setting. Given the coefficients in the direct spectral problem, one of the very first questions to answer is the determination of the (essential) support of the associated measure. An important class of coefficients is given by those that are random in the sense that each $\alpha_n$ is drawn according to some distribution, and these drawings are independent and use the same distribution for each $n$.

Such random sequences are special cases of so-called dynamically defined sequences, where the $\alpha_n$'s are determined by $\D$-valued sampling along the orbit of some map $T : \Omega \to \Omega$ that one iterates. That is, $\alpha_n = \alpha_n(\omega) = f(T^n \omega)$ for some $f : \Omega \to \D$ and $\omega \in \Omega$. If one chooses a $T$-ergodic measure $\mu$, then the general theory implies that there is a set $\Sigma \subseteq \partial \D$ such that for $\mu$-almost every $\omega \in \Omega$, the essential support of the measure on $\partial \D$ associated with $\{ \alpha_n(\omega) \}_{n \ge 0}$ is equal to $\Sigma$. Since the essential support of the measure is equal to the essential spectrum of the associated CMV matrix, this can also be rephrased as follows. The essential spectrum of a CMV matrix with dynamically defined Verblunsky coefficients is almost surely constant; see \cite[Theorem~10.16.2]{S05b}. In fact, in developing the general theory for dynamically defined Verblunsky coefficients, it is very natural to pass to a two-sided infinite setting. That is, one considers invertible maps $T$ and then considers $\{ \alpha_n(\omega) \}_{n \in \Z}$, where again $\alpha_n(\omega) = f(T^n \omega)$. The associated extended CMV matrix looks locally like a CMV matrix, but it acts in $\ell^2(\Z)$ and is two-sided infinite as well. The set $\Sigma$ discussed above turns out the be the almost sure spectrum of these extended CMV matrices, that is, these matrices have purely essential spectrum, which coincides with that of their half-line restrictions; see \cite[Theorems~10.16.1 and 10.16.2]{S05b}.

The random case arises from the general scenario by choosing a probability measure $\nu$ on $\D$, and then setting $\Omega = (\mathrm{supp} \, \nu)^\Z$, $T : \Omega \to \Omega$, $(T \omega)_n = \omega_{n+1}$, $f(\omega) = \omega_0$, and $\mu = \nu^\Z$. Thus, one is naturally interested in determining the set $\Sigma$ in this case. This problem turns out to be surprisingly challenging. There is an important analogy between the study of CMV matrices and the study of discrete Schr\"odinger operators, or more generally Jacobi matrices. Indeed, CMV matrices are canonical examples of unitary operators (with a cyclic vector), while Jacobi matrices are canonical examples of self-adjoint operators (with a cyclic vector). Moreover, the tools used in the analysis of one of these classes of matrices usually have counterparts for the other class. An overwhelming number of results have been carried over from one setting to the other. Indeed, much of \cite{S05b} is devoted to carrying over results from discrete Schr\"odinger operators to CMV matrices. Identifying the almost sure spectrum is quite easy in the Schr\"odinger case. A CMV counterpart does not yet exist to the best of our knowledge, despite there being a number of papers studying CMV matrices with random Verblunsky coefficients; compare, for example, \cite{BHJ03, GN01, T92a, T92b}. In fact we will show that the description of the almost sure spectrum in the random case that is known in complete generality in the Schr\"odinger case does \emph{not} carry over to the CMV setting in general!

We will investigate the almost sure spectrum of extended CMV matrices with random coefficients in Section~\ref{s.randomCMV} and prove several results concerning this set. A sample result is the following.

\begin{theorem}\label{t.asspositivevc}
Given a probability measure $\nu$ on $\D$ whose topological support $\mathrm{supp} \, \nu$ is contained in $[0,1]$. Then, the associated almost sure spectrum $\Sigma$ is given by
$$
\Sigma = \{ e^{i \theta} : \theta \in [-\pi, \pi), |\theta| \ge 2 \arcsin \min \mathrm{supp} \, \nu \}.
$$
\end{theorem}

Thus, for non-negative random Verblunsky coefficients, we determine the almost sure spectrum explicitly. In the absence of such a non-negativity condition, we are still able to give a description of the almost sure spectrum in terms of the spectra of suitable extended CMV matrices with periodic Verblunsky coefficients.

\bigskip

In Section~\ref{sec:app} we study the one-dimensional Ising model with (positive, or \textit{ferromagnetic}) nearest neighbor interaction, immersed in a transverse magnetic field. Indeed, this model is the simplest among Ising-type models that already in one dimension presents interesting challenges. In this paper we focus on the distribution of the Lee-Yang zeros in the thermodynamic limit, where the neighbor interaction is given by an arbitrary bounded sequence.

Let $\Lambda_N:=\set{\pm 1}^N$, with $N\in\Z_+$. The one-dimensional ferromagnetic nearest-neighbor Ising model on the lattice of length $N$ with constant field is defined as
$$
E(\sigma) := -\frac{1}{k_B\tau}\sum_{j=1}^NJ_j\sigma_j\sigma_{j+1}+H\sigma_j
$$
with $\tau \in (0,\infty)$ being the temperature, $k_B > 0$ is the Boltzmann constant (often factored into $\tau$), $\set{J_j}_{1\leq j \leq N}$, $J_j > 0$ for all $j$, are the interaction couplings, $H \neq 0$ is the magnetic field, and $\sigma=(\sigma_1, \dots, \sigma_N)\in\Lambda_N$. We work with periodic boundary conditions $\sigma_1=\sigma_{N+1}$. Also notice that the model is defined on the positive side of the integer lattice. One could also define the two-sided (and, without loss of generality, symmetric) model by considering it on the lattice $[-N, N]\subset\Z$. There is no difference between the two in the finite case.

Define the \textit{partition function} as
$$
Z^{(N)}_\tau := \sum_{\sigma\in\Lambda_N}e^{-E(\sigma)}.
$$

Let us also introduce the following change of variables.
$$
\beta_j = \exp\left(\frac{2J_j}{k_B\tau}\right)\hspace{2mm}\text{ and }\hspace{2mm}h:=\exp\left(\frac{2H}{k_B\tau}\right).
$$
Obviously $Z^{(N)}_\tau$ can be written in terms $\beta_j$'s and $h$, in which case it becomes a Laurent polynomial in $h$.

All the thermodynamic properties of the model can be derived from $Z^{(N)}_\tau$. In particular, phase transitions can be studied in terms of the regularity of $\log Z^{(N)}_\tau$. Let us consider $h\in\C$ by taking $H\in\C$. For all $N$ finite, $Z^{(N)}_\tau$ is clearly analytic in $h$, and $\log Z^{(N)}_\tau$ is analytic for $h\in \R_+$. In this case no phase transitions occur. Thus a relevant question is whether phase transitions occur in the \textit{thermodynamic limit} $N\rightarrow\infty$. That is, we are interested to know whether the limit $\lim_{N\rightarrow\infty}\frac{1}{N}\log Z_\tau^{(N)}$ exists, and if it does, what is its regularity in the variable $\log h$. For example, if the limit exists but its $k$th derivative with respect to $\log h$ does not exist or is discontinuous, we say that the model undergoes a $k$th-order phase transition in the thermodynamic limit (the physically relevant case being $h\in (0,\infty)$). In \cite{LY52a, LY52b}, Lee and Yang related the problem to the distribution of zeros of $Z^{(N)}_\tau(h)$ in the thermodynamic limit. Namely, they showed that the zeros lie on the unit circle for all $N$, and that phase transitions in the thermodynamic limit occur if and only if the point $1$ is an accumulation point of the zeros (separating the system into two phases).

On the other hand, the authors of \cite{DMY13} succeeded in relating the zeros of $Z_\tau^{(N)}$ in the thermodynamic limit to the spectrum of an associated CMV matrix. Using this connection, we are able to prove in this paper that, in complete generality, the zeros do not accumulate at 1, thus precluding any phase transitions. More precisely, we will show the following result in Section~\ref{sec:app}.

\begin{theorem}\label{t.spectralgap}
For every bounded sequence $\set{J_j}$ of ferromagnetic couplings and every $\tau > 0$, there is a neighborhood $\mathcal{N}$ of $1$ such that for every $N \in\Z_+$, $Z_\tau^{(N)}$ does not have any zeros in $\mathcal{N}$.
\end{theorem}

In fact, the neighborhood $\mathcal{N}$ is explicitly given in terms of $\tau$ and $\sup_j J_j$. To the best of our knowledge, while expected in the community, this result hitherto has not been rigorously established.

When the interaction couplings are dynamically defined, we also relate the limit of the zero counting measures to the density of states measure for the associated CMV matrix. Finally, we establish a gap labeling scheme for the zeros in the thermodynamic limit. The latter result has been postulated by a few authors, based primarily on numerical evidence (e.g. \cite{BGP95, BG01}). The conjectures there were guided by the gap labeling available for Schr\"odinger operators; however, an appropriate mechanism of relating the thermodynamic limit of the zeros to spectra of operators on infinite-dimensional Hilbert spaces in a way that would be suitable for carrying over such results as gap labeling had not been suggested until \cite{DMY13}.

\bigskip

As a matter of fact, the questions we are interested in, concerning both the almost sure random spectrum and the one-dimensional Ising model, are intimately related to the same basic concept, namely the uniform hyperbolicity of Szeg\H{o} cocycles over the shift transformation. Specifically, this means the following. If we iterate the recursion \eqref{e.tmbasic}, we are naturally led to a study of products of matrices of the form
\begin{equation}\label{e.tmonestep}
A(\alpha,z) = \frac 1\rho \begin{pmatrix} z & - \bar\alpha \\ - \alpha z & 1 \end{pmatrix},
\end{equation}
where, as usual, $\rho = (1-|\alpha|^2)^{1/2}$. Here, $z$ will be held fixed in the product, while $\alpha$ may vary from factor to factor. The goal is to show that the norm of this product, $\| A(\alpha_n,z) \cdots A(\alpha_1,z) \|$, will satisfy a lower bound that is exponentially growing in $n$ and uniform in the choice of the $\alpha_n$, restricted to some subset $S$ of $\D$. In this case, we say that the $z$ and $S$ in question lead to uniformly hyperbolic behavior. This setting can be rephrased in dynamical systems language, which (we will do in Section~\ref{s.randomCMV} and which) justifies the terminology. In Section~\ref{s.uh} we will prove uniformly hyperbolic behavior for any set $S$ that is contained in $(0,1)$ and bounded away from $0$, and corresponding $z$'s that belong to some arc in $\partial \D$ around $1$ that depends on the distance of $S$ from $0$. This will be precisely what we need in our study of the Ising model in Section~\ref{sec:app}, and it will allow us to identify the almost sure random spectrum in Section~\ref{s.randomCMV}, provided that $\nu$ is supported in some set $S$ of the form above.
\bigskip

\noindent\textbf{Acknowledgement.} The authors are grateful to Michael Baake and Uwe Grimm for helpful discussions on the Ising model.

\section{Uniform Hyperbolicity for Szeg\H{o} Cocycles with Positive Coefficients}\label{s.uh}

\begin{defi}\label{def:spectral-gap}
Given $\alpha \in (0,1)$, we set $R_\alpha = \{ e^{i \theta} : - 2 \arcsin \alpha < \theta < 2 \arcsin \alpha \}$.
\end{defi}

Fix $A \in (0,1)$. We wish to work with Verblunsky coefficients $\alpha \in [A,1)$. Let $z = w^2 \in R_A  \subseteq  \partial \D$. We can choose the root $w$ with $\arg w \in (- \arcsin A, \arcsin A)$, so that
\begin{equation}\label{e.RewImw}
\Re w > \sqrt{1-A^2}, \quad \lvert \Im w \rvert < A.
\end{equation}
Let
\[
C = \sqrt{\frac{1+A}{1-A}}, \quad \kappa = C \Re w - \lvert \Im w \rvert.
\]
It is easy to check that
\[
\kappa > 1.
\]
The usual transfer matrices \eqref{e.tmonestep} can be divided by $\sqrt z = w$ and conjugated by
\[
U = \frac 1{\sqrt 2} \begin{pmatrix} 1 & 1 \\ -i  & i \end{pmatrix}.
\]
This gives matrices of the form (using $\alpha \in \mathbb{R}$)
\[
B(\alpha, w)
=
w^{-1} U A(\alpha, w^2) U^{-1}
=
\frac 1\rho \begin{pmatrix} (1-\alpha) \Re w & -(1-\alpha) \Im w \\   (1+\alpha) \Im w  &  (1+\alpha) \Re w \end{pmatrix}.
\]
Uniform hyperbolicity follows from the following statement.

\begin{lemma}\label{l.uh}
Suppose $A \in (0,1)$ and $z = w^2 \in R_A$. Let
\[
\begin{pmatrix} \tilde x \\ \tilde y \end{pmatrix} = B(\alpha, w) \begin{pmatrix} x \\ y \end{pmatrix}.
\]
\begin{enumerate}

\item[{\rm (a)}] If $y > C \lvert x \rvert$, then $\tilde y > C \lvert \tilde x \rvert$.

\item[{\rm (b)}] If $y > C \lvert x \rvert$, then $\tilde y > \kappa y$.

\end{enumerate}
\end{lemma}

\begin{proof}
(a) We begin by noticing that the expression
\[
f(\alpha) = 2 \alpha \Re w - \left[ C + \frac 1C + \alpha \left(\frac 1C - C \right)  \right] \lvert \Im w \rvert
\]
is an increasing function of $\alpha$, so
\begin{align*}
f(\alpha) \ge f(A) 
= 2 A \Re w - 2 \sqrt{1-A^2} \lvert \Im w \rvert   > 0
\end{align*}
by \eqref{e.RewImw}. By rearranging terms, we can write the inequality $f(\alpha) > 0$ in the form
\[
 (1+\alpha) \Re w - C (1-\alpha) \lvert \Im w \rvert > \frac 1C \left( (1+\alpha) \lvert \Im w\rvert + C (1-\alpha) \Re w   \right)
\]
Multiplying this inequality by the inequality $y > C \lvert x\rvert$ we get
\[
\left( (1+\alpha) \Re w - C (1-\alpha) \lvert \Im w \rvert \right) y > \left( (1+\alpha) \lvert \Im w\rvert + C (1-\alpha) \Re w   \right) \lvert x\rvert
\]
which can, again, be rearranged into the form
\[
- (1+\alpha)\lvert\Im w \rvert \lvert x \rvert  + (1+\alpha) \Re w \, y  > C \left( (1-\alpha) \Re w \lvert x \rvert + (1-\alpha) \lvert \Im w \rvert y \right).
\]
It is clear that the left hand side of this inequality is smaller or equal to $\rho \tilde y$, and the right hand side is larger or equal to $C \rho \lvert \tilde x \rvert$, so we conclude $\tilde y > C \lvert \tilde x\rvert$.

(b) By the triangle inequality and by $\lvert x \rvert < \frac 1C y$,
\begin{equation}\label{e.tildeyineq}
\tilde y > \frac 1\rho \left(  (1+\alpha) \Re w \, y - (1+\alpha)\lvert\Im w \rvert \lvert x \rvert  \right)> \frac {1+\alpha}\rho \left( \Re w  - \lvert\Im w \rvert \frac 1C \right) y.
\end{equation}
Since
\[
\frac{1+\alpha}\rho = \sqrt{\frac{1+\alpha}{1-\alpha}} \ge \sqrt{\frac{1+A}{1-A}} = C
\]
and
\[
\Re w - \lvert \Im w \rvert \frac 1C > \sqrt{1-A^2} - A \sqrt{\frac{1-A}{1+A}} = (1+A)\frac 1C - A \frac 1C = \frac 1C
\]
the inequality \eqref{e.tildeyineq} implies $\tilde y > \kappa y$, where
\[
\kappa = C \Re w - \lvert \Im w \rvert > 1. \qedhere
\]
\end{proof}

Lemma~\ref{l.uh} immediately implies that arbitrary products formed from matrices of the type $A(\alpha,z)$ with $\alpha \in [A,1)$ and $z = w^2$ satisfying \eqref{e.RewImw} have norm that grows exponentially in the number of factors, uniformly in the choice of the parameters:

\begin{theorem}\label{t.uniformgrowth}
For every $A \in (0,1)$, there exist $C > 0$ and $\lambda > 1$ such that for every sequence $\{ \alpha_n \}_{n \in \Z_+} \subset [A,1)$ and every $z \in R_A$, we have
$$
\| A(\alpha_n,z) \times \cdots \times A(\alpha_1,z) \| \ge C \lambda^n
$$
for every $n \in \Z_+$.
\end{theorem}

\section{The Almost Sure Spectrum of Random Extended CMV Matrices}\label{s.randomCMV}

Fix $0 < A < B < 1$ and consider the compact interval $[A,B]$. Taking the infinite product $\Omega = [A,B]^\Z$ and equipping it with the product topology, $\Omega$ is compact as well. The shift transformation $T : \Omega \to \Omega$ is given by $(T \omega)_n = \omega_{n+1}$. Let $z \in \partial \D$ and consider the map
$$
A_z(\cdot) : \Omega \to U(1,1), \quad \omega \mapsto (1 - |\omega_0|^2)^{-1/2} \begin{pmatrix} z & - \bar\omega_0 \\ - \omega_0 z & 1 \end{pmatrix}.
$$
With this map, we define the associated Szeg\H{o} cocycle
$$
(T, A_z) : \Omega \times \C^2 \to \Omega \times \C^2, \quad (\omega,v) \mapsto (T \omega, A_z(\omega) v),
$$
which is an extension of the shift transformation in the base $\Omega$. Notice that $(T, A_z)$ is invertible and, for $n \in \Z$, we can write the iterate $(T, A_z)^n$ as $(T^n, A_z^n)$ with a suitable map $A_z^n : \Omega \to U(1,1)$. For example, for $n \ge 1$, we have $A_z^n(\omega) = A_z(T^{n-1} \omega) \cdots A_z(\omega)$.

\begin{theorem}\label{t.uh}
The cocycle $(T, A_z)$ is uniformly hyperbolic for every $z \in R_A$. That is, for each $z \in R_A$, each of the following equivalent statements holds.
\begin{itemize}

\item[{\rm (a)}] There exist constants $\lambda > 1$ and $C>0$ {\rm (}which \emph{a priori} depend on $z${\rm )} so that $ \|A_z^n (\omega) \| \geq C\lambda^{|n|} $ for all $n \in \Z$ and all $\omega \in \Omega$.

\item[{\rm (b)}] $(T,A_z)$ admits a continuous invariant exponential splitting. That is to say, there are continuous maps $\Lambda_z^s,\Lambda_z^u : \Omega \to \C\PP^1$ and constants $c>0,L>1$ so that the following two conditions hold. First,
$$
A_z(\omega) \cdot \Lambda^s_z(\omega)
=
\Lambda^s_z(T\omega)
\quad
\text{and}
\quad
A_z(\omega) \cdot \Lambda^u_z(\omega)
=
\Lambda^u_z(T\omega)
$$
for every $\omega$.  Second,
$$
\| A_z^n v_s \| \leq cL^{-n},
\quad
\| A_z^{-n} v_u \| \leq cL^{-n}
$$
for all $ v_s \in \Lambda^s_z(\omega), v_u \in \Lambda_z^u(\omega) $ and every $n \geq 0$.

\item[\rm (c)] $(T,A_z)$ does not enjoy a Sacker-Sell type orbit.  That is to say, for any $\omega \in \Omega$ and any $v \in \C^2$ such that $\|v\| = 1$, there exists an $n \in \Z$ for which $\| A_z^n(\omega) v \| > 1$.

\end{itemize}
\end{theorem}

\begin{proof}
Part (a) is an immediate consequence of Theorem~\ref{t.uniformgrowth}. More precisely, for $n \in \Z_+$, the statement is precisely the one given in Theorem~\ref{t.uniformgrowth}, while for $n \in \Z_-$ one notes that \begin{align*}
A^n_z(\omega) & = A_z(T^n\omega)^{-1} \cdots A_z(T^{-1}\omega)^{-1} \\
& = \left( A_z(T^{-1}\omega) \cdots A_z(T^n\omega) \right)^{-1} \\
& = \left( A(\omega_{-1},z) \cdots A(\omega_{n},z) \right)^{-1},
\end{align*}
and that $\|A^n_z(\omega)^{-1}\| = \|A^n_z(\omega)\|$ since $A^n_z(\omega) \in U(1,1)$, and then applies Theorem~\ref{t.uniformgrowth}. The equivalence of (a),  (b), and (c) is well-known; compare \cite[Theorem~A]{BG}, \cite[Proposition~2]{Y04}, and \cite{Z}. See also \cite{J} \cite{sacksell1}, \cite{sacksell2}, \cite{selgrade}.
\end{proof}

With this result in hand, we can now describe the almost sure spectrum of random extended CMV matrices with positive Verblunsky coefficients. Suppose $\nu$ is a probability measure with topological support $\mathrm{supp} \, \nu \subset (0,1)$. Set $\mu = \nu^\Z$ and, for $\omega \in \mathrm{supp} \, \mu = (\mathrm{supp} \, \nu)^\Z$, denote by $\mathcal{E}_\omega$ the extended CMV matrix with coefficients $\{ \omega_n \}_{n \in \Z}$. By the general theory of ergodic extended CMV matrices there exists a closed set $\Sigma_\mu \subset \partial \D$ such that $\sigma(\mathcal{E}_\omega) = \Sigma_\mu$ for $\mu$-almost every $\omega$. Namely, modulo minor adjustments this follows from \cite[Theorem~10.16.2]{S05b} (the essential spectrum in the one-sided case and the spectrum in the two-sided case coincide).

\begin{theorem}\label{t.asspositivevc2}
We have $\Sigma_\mu = \partial \D \setminus R_{\min \mathrm{supp} \, \nu}$.
\end{theorem}

\begin{proof}
By Theorem~\ref{t.uh}, the cocycle $(T, A_z)$ is uniformly hyperbolic for every $z \in R_{\min \mathrm{supp} \, \nu}$. It follows that that $R_{\min \mathrm{supp} \, \nu}$ is fully contained in $\partial \D \setminus \Sigma_\mu$ by \cite[Theorems~2.6 and 5.1]{GJ96}. Note that the results from \cite{GJ96} apply since $\mathrm{supp} \, \nu$ is a compact subset of $(0,1)$.

Conversely, note that for every $\alpha \in \mathrm{supp} \, \nu$, the extended CMV matrix with constant coefficients given by $\alpha$ has spectrum $\partial \D \setminus R_{\alpha}$; compare \cite[Example~1.6.12]{S05}. Using trial vectors, one can see that this set must be contained in $\Sigma_\mu$. (We will give  detailed arguments on how to reach this conclusion in a more general case in the proof of Theorem~\ref{thm:as-spectr} below.) Choosing in particular $a = \min \mathrm{supp} \, \nu$, it follows that $\partial \D \setminus R_{\min \mathrm{supp} \, \nu}$ must be contained in $\Sigma_\mu$.
\end{proof}

Theorem~\ref{t.asspositivevc}, stated in Section~\ref{s.intro}, is an immediate consequence of Theorem~\ref{t.asspositivevc2}. We just need to note that if $\min \mathrm{supp} \, \nu = 0$, then $\Sigma_\mu = \partial \D$ follows immediately from \cite[Example~1.6.12]{S05} and the argument given in the proof above.

\begin{remark}
Note that the almost sure spectrum in the CMV case is obtained in a different way than in the Schr\"odinger case. Recall how this set may be described in the latter case. If one considers a probability measure $\nu$ with compact topological support $\mathrm{supp} \, \nu \subset \R$, sets $\mu = \nu^\Z$, and considers, for $\omega \in \mathrm{supp} \, \mu = (\mathrm{supp} \, \nu)^\Z$, the operator
$$
[H_\omega \psi](n) = \psi(n+1) + \psi(n-1) + \omega_n \psi(n)
$$
in $\ell^2(\Z)$, then by the general theory of ergodic Schr\"odinger operators there exists a compact set $\Sigma_\mu \subset \R$ such that $\sigma(H_\omega) = \Sigma_\mu$ for $\mu$-almost every $\omega$. We have
\begin{equation}\label{e.schaespec}
\Sigma_\mu = [-2,2] + \mathrm{supp} \, \nu \, ;
\end{equation}
see, for example, \cite[Theorem~3.9]{K08}. The identity \eqref{e.schaespec} says that the almost sure spectrum is the set sum of the spectrum of the Laplacian and the almost sure spectrum of the potential. It may also be interpreted as follows,
\begin{equation}\label{e.schaespec2}
\Sigma_\mu = \bigcup_{a \in \mathrm{supp} \, \nu} \sigma(\Delta + a).
\end{equation}
That is, the almost sure spectrum is the union of the spectra of all Schr\"odinger operators with constant potential, where the constant runs through the topological support of the single-site measure $\nu$. This characterization does not extend to the CMV case {\rm (}and hence there is no obvious analog of \eqref{e.schaespec}{\rm )}! To see this, consider a measure $\nu$ with $\mathrm{supp} \, \nu = \{ \alpha, - \alpha \}$ for some small $\alpha > 0$ and the associated family $\{ \mathcal{E}_\omega \}_{\omega \in \mathrm{supp} \, \mu}$ of extended CMV matrices. There are two constant sequences in $\mathrm{supp} \, \mu$, and the spectrum of each of them is equal to $\partial \D \setminus R_\alpha$ by \cite[Example~1.6.12]{S05}. The analog of \eqref{e.schaespec2} would therefore suggest that $\Sigma_\mu = \partial \D \setminus R_\alpha$. However, the spectrum of the $2$-periodic extended CMV matrix with coefficients $\alpha, - \alpha$ on each period has spectrum given by the reflection of $\partial \D \setminus R_\alpha$ about the imaginary axis {\rm (}or rather the set rotated about the origin by $\pi${\rm )}; compare \cite[(3.2.6)]{S05}. A trial function argument shows that the latter set must be contained in $\Sigma_\mu$. Thus, if $\alpha$ is sufficiently small, it follows that $\Sigma_\mu = \partial \D$. This shows that determining the almost sure spectrum in the random case is different for CMV matrices in comparison with Schr\"odinger operators.
\end{remark}

A weak replacement of \eqref{e.schaespec2} in the CMV case is given in the next theorem. It is shown that, while considering the constant sequences drawn from the support of $\nu$ is not enough to determine $\Sigma_\mu$ as discussed in the previous remark, it does suffice to consider the \emph{periodic} sequences drawn from $\supp \, \nu$. This is a result similar to \cite[Theorem~4]{KirschMart92}.

\begin{theorem}\label{thm:as-spectr}
Let $\nu$ be a compactly supported probability measure on $\D$.  Form $\Omega = \supp(\nu)^{\Z}$ and the corresponding family of random two-sided CMV matrices $ (\E_\omega)_{\omega \in \Omega} $ as before.  Then
\begin{equation} \label{eq:as:spectrum}
\Sigma_\mu
=
\overline{
\bigcup_{p \in \Z_+} \bigcup_{\alpha_1, \ldots, \alpha_p \in \supp(\nu)} \sigma(\mathcal{E}_{\alpha_1,\ldots,\alpha_p})},
\end{equation}
where $\mathcal{E}_{\alpha_1,\ldots,\alpha_p}$ denotes a periodic two-sided CMV matrix with Verblunsky coefficients $\alpha_1,\ldots,\alpha_p$.
\end{theorem}

\begin{proof}
Let $\Omega_0$ denote a set of full $\mu$ measure such that $\sigma(\E_\omega) = \Sigma_\mu$ for every $\omega \in \Omega_0$. Enlarging $\Omega_0$ if necessary, we may assume that $\Omega_0$ is shift-invariant.  After doing this, we may remove a set of zero $\mu$-measure to ensure that the following holds: given any $\varepsilon > 0 $ and any finite sequence $\alpha_1,\ldots,\alpha_k \in \supp(\nu) $, there exists $\omega' \in \Omega_0$ such that $ | \omega'_j - \alpha_j | < \varepsilon $ for all $1 \leq j \leq k$.

The inclusion ``$\supseteq$'' follows from a trial function argument.  Specifically, if $z$ is in the spectrum of $\E_{\alpha_1,\ldots,\alpha_p}$, then there is a Floquet solution $\phi$ so that $\E_{\alpha_1,\ldots,\alpha_p}\phi = z \phi$ and $|\phi_n|$  is bounded; compare \cite[Section~11.2]{S05b}.  By assumption, for each $N \in \Z_+$, there is some $\omega^N \in \Omega_0$ with
$$
|\omega^N_{j+kp} - \alpha_j| < 2^{-N},
\quad
\text{for all } 1 \leq j \leq p, \; 0 \leq k \leq N-1.
$$
Now, take $\phi_N = \phi \chi_{[1,pN]}$. By Floquet theory, we have
\begin{equation} \label{eq:per.weyl}
\lim_{N \to \infty} \frac{\| (\E_{\omega^N}-z) \phi_N\|}{\|\phi_N\|} = 0.
\end{equation}
From \eqref{eq:per.weyl}, we deduce $z \in \Sigma_\mu $, since
\begin{align}\label{eq:dist-to-spectr}
\mathrm{dist}(z,\sigma(\E))
=
\inf_{\| \psi \| = 1} \| (\E - z) \psi \|.
\end{align}
(The identity \eqref{eq:dist-to-spectr} is a standard consequence of the spectral theorem for unitary operators.)

To prove the reverse inclusion, fix $\omega \in \Omega_0$ and let $\mathcal{G}_\omega$ denote the set of generalized eigenvalues of $\mathcal{E}_\omega$, that is, the set of $z$ for which there exists a nonzero polynomially bounded $\phi:\Z \to \C$ for which $\E_\omega \phi = z \phi$. By a standard result,
\begin{equation} \label{eq:schnol:cmv}
\overline{\mathcal{G}_{\omega}} = \Sigma_\mu.
\end{equation}
Specifically, \eqref{eq:schnol:cmv} follows from trivial modifications to the proof of \cite[Theorem~7.1]{K08}; compare \cite[Corollary~2.11]{CFKS}. Given $z \in \mathcal{G}_\omega$, let $\phi$ be a corresponding generalized eigenfunction. For each $N \in \Z_+$, take $\phi_N = \phi \cdot \chi_{[-N,N]}$. Since $\phi_N$ is polynomially bounded,
$$
\liminf_{N \to \infty} \frac{\| (\mathcal{E}_\omega - z) \phi_N \|}{\| \phi_N \|}
=
0.
$$
 But then, if $\mathcal{E}_{\omega,N} = \mathcal{E}_{\omega_{-N},\ldots,\omega_N}$, then
$$
\liminf_{N \to \infty}
\frac{\| (\mathcal{E}_{\omega,N} - z) \phi_N \|}{\| \phi_N \|}
=
0,
$$
which proves that $z$ is in the right hand side of \eqref{eq:as:spectrum}.
\end{proof}

\begin{remark}
In general, the union on the RHS of \eqref{eq:as:spectrum} is rather large and intractable -- one may hope that not all periodic spectra are necessary.  One could ask whether it suffices to only consider spectra of matrices with period one or two.  However, this is insufficient -- for example, take $\alpha = .6,\beta = .9 i$ and let $\nu$ be a Bernoulli distribution supported on $\{\alpha,\beta\}$.  With $ z = e^{1.1i}$, numerical computations show that
$$
z \notin \sigma(\E_\alpha) \cup \sigma(\E_\beta) \cup \sigma(\E_{\alpha,\beta}),
$$
but $z \in \sigma(\E_{\alpha,\alpha,\alpha,\beta})$.  In particular, in this case, the spectra of period two matrices do not cover $\Sigma_\mu$.
\end{remark}

As an immediate corollary of the above theorem, we can relate the almost sure spectrum $\Sigma_\mu$ to the spectra of periodic approximations of $\E_\omega$ for almost every $\omega\in\Omega$ as follows.

\begin{coro}\label{coro:as-spectr}
For any two sequences of natural numbers $\set{l_k}$ and $\set{r_k}$ with $l_k\nearrow\infty$ and $r_k\nearrow\infty$ as $k\nearrow\infty$ we have the following. For $\mu$ almost all $\omega\in\Omega$,
$$
\Sigma_\mu = \overline{\bigcup_k \sigma(\E_{\omega, (l_k, r_k)})},
$$
where $\E_{\omega, (l_k, r_k)}$ is the periodic two-sided CMV matrix formed from the periodic Verblunsky coefficients with the unit cell $(\omega_{-l_k}, \dots, \omega_{r_k})$.
\end{coro}

\begin{proof}
Fix any $\omega\in\Omega$ and $z\in \sigma(\E_\omega)$. Let us first prove that for any $\varepsilon > 0$, there exists $k\in\N$ such that $\mathrm{dist}(z, \sigma(\E_{\omega, (l_k, r_k)})) < \varepsilon$.

Due to \eqref{eq:dist-to-spectr} we can fix a sequence of unit vectors $\set{\phi_n}\subset\ell^2$ such that $\norm{(\E_\omega - z)\phi_n}\rightarrow 0$ as $n\rightarrow\infty$. Notice that $\E_{\omega, (l_k, r_k)}$ converges to $\E_\omega$ in the strong operator topology. Therefore, there exists $n_0\in \N$ such that for all $k$ sufficiently large, $\norm{(\E_{\omega, (l_k, r_k)}-z)\phi_{n_0}} < \varepsilon$. We obtain the desired bound $\mathrm{dist}(z, \sigma(\E_{\omega, (l_k, r_k)}))<\varepsilon$ by \eqref{eq:dist-to-spectr}.

Since for $\mu$-almost all $\omega$, we have $\sigma(\E_\omega)=\Sigma_\mu$, by Theorem~\ref{thm:as-spectr} we obtain, for $\mu$-almost all $\omega \in \Omega$,
$$
\Sigma_\mu \subseteq \overline{\bigcup_k \sigma(\E_{\omega, (l_k, r_k)})}\subset \overline{
\bigcup_{p \in \Z_+} \bigcup_{\alpha_1, \ldots, \alpha_p \in \supp(\nu)} \sigma(\mathcal{E}_{\alpha_1,\ldots,\alpha_p})} = \Sigma_\mu,
$$
and the corollary is proved.
\end{proof}

As a consequence of Corollary \ref{coro:as-spectr} we can now approximate the almost sure spectrum by the zeros of the discriminant. This will play an important role in Section \ref{sec:app}.

In what follows, by $\mathrm{dist}_H(\cdot, \cdot)$ we denote the Hausdorff distance on $\partial \mathbb{D}$ (not to be confused with $\mathrm{dist}(\cdot, \cdot)$ that was used in Theorem \ref{thm:as-spectr} and Corollary \ref{coro:as-spectr}). With the notation from Corollary \ref{coro:as-spectr}, let us denote by $\Delta^{(k)}_\omega(z)$ the discriminant over $[-l_k, r_k]\cap \Z$ (i.e. the trace of the product $\prod_{j=r_k}^{-l_k}A(\omega_j, z)$, where $A(\omega_j, z)$ is the Szeg\H{o} cocycle). Denote the set of zeros of $\Delta^{(k)}_\omega$ by $\mathcal{Z}_\omega^{(k)}$.

\begin{theorem}\label{thm:approx}
For $\mu$ almost every $\omega\in\Omega$ we have $\lim_{k\rightarrow\infty}\mathrm{dist}_H(\mathcal{Z}_\omega^{(k)}, \Sigma_\mu)= 0.$
\end{theorem}

\begin{proof}
By Corollary \ref{coro:as-spectr}, for $\mu$-almost every $\omega\in\Omega$ and any $k$, we have $\sigma(\E_{\omega, (l_k, r_k)})\subset \Sigma_\mu$. On the other hand, for all $k$, the spectrum of $\E_{\omega, (l_k, r_k)}$ is a union of compact arcs in $\partial\mathbb{D}$ containing $\mathcal{Z}_\omega^{(k)}$ (these arcs may intersect only at the endpoints), each of which contains precisely one point from $\mathcal{Z}_\omega^{(k)}$ in its interior \cite[Theorem 11.1.1]{S05b}. Thus it remains to prove that no point of $\Sigma_\mu$ remains (uniformly in $k$) away from $\sigma(\E_{\omega, (l_k, r_k)})$.

Fix $\varepsilon > 0$. Let $\set{C_j}_{j=1,\dots,m}$ be a finite open cover of $\Sigma_\mu$ such that for all $j$, the diameter $\mathrm{diam}(C_j) < \varepsilon/4$. Take a typical $\omega\in\Omega$ (i.e. such that $\sigma(\E_\omega) = \Sigma_\mu$). Fix $j\in\set{1, \dots, m}$ and $z\in C_j$. From the proof of Corollary \ref{coro:as-spectr}, we know that $\mathrm{dist}(z, \sigma(\E_{\omega, (l_k, r_k)})) < \varepsilon/4$ for all sufficiently large $k$.  On the other hand, Lemma~5 in \cite{O12} guarantees that the length of each such arc is bounded by $2\pi/(l_k+r_k)$, which can be made smaller than $\varepsilon/4$ for all $k$ sufficiently large.
\end{proof}

We will need the following result in Section~\ref{sec:app}; let us record it here for completeness as a theorem.

\begin{theorem}\label{thm:approx-periodic}
With the notation from Corollary \ref{coro:as-spectr}, we have the following. If $\omega\in\Omega$ is periodic, then for any choice of sequences of natural numbers $l_k$ and $r_k$ as in Corollary \ref{coro:as-spectr}, we have $\lim_{k\rightarrow\infty}\mathrm{dist}_H(\mathcal{Z}_\omega^{(k)}, \sigma(\E_\omega))=0$.
\end{theorem}

\begin{proof}
That no point of $\sigma(\E_\omega)$ remains (uniformly in $k$) away from $\sigma(\E_{\omega, (l_k, r_k)})$ follows by the same arguments as those in the proof of Theorem \ref{thm:approx}.

Next we prove that the spectra $\sigma(\E_{\omega, (l_k, r_k)})$ do not contain points which remain (as $k\nearrow\infty$) away from $\sigma(\E_\omega)$. Assume to the contrary that there exists $\delta > 0$ and, for every $k$, there exists $z_k\in \sigma(\E_{\omega, (l_{k}, r_{k})})$ such that $\mathrm{dist}(z, \sigma(\E_\omega)) > \delta$ (actually, to be pedantic, we should be taking a subsequence $\set{k_j}$, but that would only reduce to considering the new sequences $l_{k_j}$ and $r_{k_j}$ instead of $l_k$ and $r_k$). It follows that for all $k$, $z_k$ belongs to a spectral gap of $\sigma(\E_\omega)$. Without loss of generality (after passing to a subsequence if necessary) we may assume that there exists a compact subinterval $J$ of a spectral gap of $\sigma(\E_\omega)$ with its boundary points at least distance $\delta$ away from the boundary of the gap, such that all the $z_k$'s belong to $J$.

Let us assume that $\omega$ is periodic with period $p$. For any $j$, let us write $\Delta_j(z)$ for the discriminant over $[1, j]\cap\Z$. There exists $M > 2$ such that for all $z\in J$, $\abs{\Delta_p(z)} > M$ (see, e.g., \cite[Theorems~11.1.1 and 11.1.2]{S05b}). It follows that with $k=np$, $n\in\Z_+$, $\abs{\Delta_k(z)}$ is exponentially large in $n$. Thus for general and large $k$, the discriminant evaluated at $z\in J$ over $[l_k, r_k]\cap \Z$ is larger than 2 in absolute value, precluding containment of $z$ in $\sigma(E_{\omega, (l_k, r_k)})$.
\end{proof}

\section{An Application to the Ising Model}\label{sec:app}

  In all of what follows, we assume the notation from Section \ref{s.intro}, where the Ising model was introduced.

It turns out that the zeros of $Z^{(N)}_\tau(h)$ are precisely the zeros of the discriminant $\Delta_N(h):=\mathrm{Tr}\prod_{j=N}^1 A(\alpha_j, h)$, with Verblunsky coefficients given by $\alpha_j=1/\beta_j$; see \cite{DMY13} for details.

Notice that for all $j$, $\alpha_j = 1/\beta_j$ belongs to $(0, 1)$. Furthermore, if the sequence $\set{J_j}$ is bounded, then the Verblunsky coefficients all lie away from zero. Thus, by the results of Section \ref{s.uh}, at least in measure, the zeros of $Z_\tau^{(N)}$ in the thermodynamic limit all lie away from $1$. However, it may happen that topologically the zeros still accumulate somewhere in the spectral gap (perhaps even at $1$). The following theorem precludes this.

\begin{theorem}\label{thm:tdlimit-gap}
With a choice of {\rm (}bounded{\rm )} ferromagnetic couplings $\set{J_j}$, $\tau > 0$, and the associated Verblunsky coefficients $\alpha_j = 1/\beta_j$, let $\alpha = \inf_{j}\set{\alpha_j}$ and $R_\alpha$ the spectral gap as in Definition \ref{def:spectral-gap}. Then for any $N \in \Z_+$, $Z_\tau^{(N)}$ does not have a zero in $R_\alpha$.
\end{theorem}

\begin{proof}
This follows immediately from the fact that the zeros of $Z_\tau^{(N)}$ coincide with those of $\Delta_N$, and an application of Theorem~\ref{t.uniformgrowth}. Namely, if for some $N \in \Z_+$, $\Delta_N$ has a zero $z$ in $R_\alpha$, then the matrix $A(\alpha_N, z) \cdots A(\alpha_1,z)$ is elliptic, and hence its powers are bounded. This contradicts the statement in Theorem~\ref{t.uniformgrowth} for the $z$ in question and the periodic sequence $\alpha_1, \ldots, \alpha_N, \alpha_1, \ldots, \alpha_N, \alpha_1, \ldots$.
\end{proof}

Theorem~\ref{t.spectralgap}, stated in Section~\ref{s.intro}, readily follows from Theorem~\ref{thm:tdlimit-gap}.

\bigskip


The following theorem shows that the zeros of the partition function in the thermodynamic limit accumulate on the spectrum of the associated two-sided CMV matrix, independently of the way the limit is taken.

\begin{theorem}\label{thm:tdlimit-top-1}
Fix $K > 0$. Let $\hat{\Omega}=(0, K]^{\Z}$ {\rm (}we equip $\hat{\Omega}$ with the product topology induced by the standard topology on $(0, K]${\rm )}. Assume that $\nu$ is a probability measure supported on $(0, K]$, and $\hat{\nu}$ is the induced product measure on $\hat{\Omega}$. Then for $\hat{\nu}$ almost every sequence of nearest neighbor interaction couplings $J=(\dots, J_{-1}, J_0, J_1, \dots)\in\hat{\Omega}$ and any $\tau > 0$, the following holds.

Fix any two sequences $N^l_k$, $N^r_k$ of natural numbers, with $N^l_k\nearrow\infty$ and $N^r_k \nearrow \infty$ as $k \nearrow \infty$. If $Z^{(k)}_\tau$ is the partition function of the Ising model on the lattice $[-N^l_k, N^r_k]\cap\Z$ with interaction $(J_{-N^l_k}, \dots, J_{N^r_k})$ and temperature $\tau > 0$, and if $\mathcal{Z}^{(k)}$ denotes its zeros, then
\begin{align}\label{eq:tdlimit-top-1}
\lim_{k \rightarrow \infty} \mathrm{dist}_H (\mathcal{Z}^{(k)}, \E_\omega) = 0,
\end{align}
where $\mathrm{dist}_H (\cdot, \cdot)$ is the Hausdorff metric on $\partial\mathbb{D}$, and $\E_\omega$ is the extended CMV matrix with Verblunsky coefficients $\omega = e^{-2J/k_B\tau}$.
\end{theorem}

\begin{proof}
Let us define $\Omega:=[e^{-2K/k_B\tau}, 1)^{\Z}$, taken with the product topology. Observe that $\hat{\nu}$ induces, via the map $x\mapsto e^{-x}$, a measure $\mu$ on $\Omega$. Furthermore, \cite[Proposition~3.1]{DMY13} extends in an obvious way, so that the set $\mathcal{Z}^{(k)}$ coincides with the zeros of $\Delta^{(k)}(z):=\mathrm{Tr}\prod_{j=N_k^r}^{-N_k^l}A(\omega_j, z)$. Thus it is enough to prove that for $\mu$ almost every $\omega\in\Omega$, the zeros of $\Delta^{(k)}$ accumulate on the spectrum of $\E_\omega$ as $k\nearrow\infty$. But this is precisely Theorem \ref{thm:approx}.
\end{proof}

\begin{remark}
Notice that since the spectrum as a set is almost surely constant in $\omega$, Theorem \ref{thm:tdlimit-top-1} implies that the topological limit distribution of the zeros is almost surely constant in the interaction couplings.
\end{remark}

We emphasize that the thermodynamic distribution of the Lee-Yang zeros with periodic choice of $\set{J_j}$ of low period has been investigated numerically (e.g. \cite{BGP95}), and rigorously (e.g. \cite{BG01}).



As a concluding remark, we point out another advantage of our approach.  It has been repeatedly asked in the physics literature  whether one can have a notion of gap-labeling for the Ising model (see, e.g., \cite[pp.\ 288--290]{BGP95} and  \cite[p.\ 858]{BG01}).  In our approach, we may use the relationship between the Lee-Yang zeros of the partition function and the spectra of CMV matrices to do just that.  In particular, \cite[Theorem 5.6]{GJ96} establishes a gap-labeling theorem for CMV operators \`{a} la the celebrated theorem of Johnson for differential operators modelling single-particle systems in \cite{J}.  Of course, \cite{GJ96} does not phrase statements in terms of CMV operators, since the term ``CMV matrix'' had not yet been coined.  We summarize the main ingredients of Geronimo-Johnson's approach below.\footnote{We do this primarily for the benefit of the reader who is interested in the Ising model and who is not familiar with OPUC theory, and especially the OPUC theory for dynamically defined Verblunsky coefficients.}

Let $\Omega$ be a compact metric space, $T: \Omega \to \Omega$ a homeomorphism, $\mu$ a $T$-ergodic measure, $f \in C(\Omega,\D)$ a sampling function, $(\E_{\omega})_{\omega \in \Omega}$ the associated dynamically defined family of extended CMV matrices (i.e, with Verblunsky coefficients $\alpha_n(\omega) = f(T^n \omega)$), and $\Sigma_\mu$ the $\mu$-almost sure spectrum thereof. The associated one-parameter family of cocycles is $(T,A_z)(\omega, v) = (T \omega, A(f(\omega),z))$, with iterates $(T,A_z)^n = (T^n,A_z^n)$ for suitable $A_z^n : \Omega \to U(1,1)$, and with associated Lyapunov exponent
$$
L(z) = \lim_{n \to \infty} \frac 1n \int \log \| A_z^n(\omega)\| \, d\mu(\omega).
$$
The rotation number $\rho(z)$ is (essentially) given by (radial boundary values of) the harmonic conjugate of the Lyapunov exponent. Namely, there exist a constant $R$ and a measure $dk$ on $\partial \D$ such that
$$
L(z) = R + \int_{\partial \D} \log |z - w| \, dk(w);
$$
compare \cite[Theorem~2.4]{GJ96}. Then, with the Floquet exponent
$$
w(z) = R + \int_{\partial \D} \log (z - w) \, dk(w),
$$
the rotation number $\rho(z)$ is defined via the identity
$$
w(z) = L(z) + i \, \frac{\rho(z) + \mathrm{Arg} \, z}{2},
$$
initially for $|z| < 1$ and then by taking radial boundary values; compare the discussion on \cite[pp.151--152]{GJ96}.

One can show under mild additional assumptions on $T$ and/or $f$ that $\rho$ is constant precisely in gaps, i.e., in subarcs of $\partial \D \setminus \Sigma$; compare \cite[Theorem~2.6]{GJ96} and \cite[Corollary~5.5]{GJ96}.

This shows that a gap of $\Sigma$ may be naturally labeled by the (constant) value $\rho$ takes on this gap. It is therefore of interest to ask about the possible values one can get in this way, and interestingly enough, this set of possible values may be determined solely from the base dynamics (i.e., it is the same for all $f$'s). This is the content of \cite[Theorem~5.6]{GJ96}:

\begin{theorem}
There is a countable subgroup $\mathfrak{M} \subseteq \R$, which depends only on $\Omega$ and $T$, with the property that $2 \rho(z) \in \mathfrak{M}$ for each $z \in \partial \D \setminus \Sigma_\mu$.
\end{theorem}

Let us briefly describe how the countable subgroup $\mathfrak{M} \subseteq \R$ containing all gap labels may be obtained from the base dynamics. Recall that the suspension of $(\Omega,T)$ is defined by $ \widehat{\Omega} =( \Omega \times \R) / \Z $, where the action of $\Z$ on $\Omega \times \R$ is given by $ n \cdot (\omega,t) = (T^n \omega, t-n) $.  Denoting the equivalence class of $(\omega,t)$ by $[\omega,t]$, we observe that $\widehat{\Omega}$ is equipped with natural flow
$$
T^s[\omega,t] = [\omega,t+s].
$$
Thus, $\widehat{\Omega}$ is a fiber bundle with fibers homeomorphic to $\Omega$ endowed with a natural $\R$-flow in such a way that the action of the time-one map on the fibers coincides with the action of $T$ on $\Omega$.  One obtains a natural measure $\widehat{\mu}$ on $\widehat{\Omega}$ by averaging pushforwards of $\mu$ -- more precisely,
$$
\int_{\widehat{\Omega}} \! f[\omega,t] \, d\widehat{\mu}[\omega,t]
=
\int_0^1 \! \int_{\Omega} \! f[\omega,t] \, d\mu(\omega) \, dt.
$$
The subgroup $\mathfrak{M}$ arises from the approach of Schwarzmann applied to $\widehat{\Omega}$, which we briefly summarize here.  The interested reader is referred to \cite{Schwarz} for a more detailed exposition.  Let $H = \check{H}^1(\widehat{\Omega},\R)$ denote the first \v{C}ech cohomology group of $\widehat{\Omega}$ with real coefficients.  Any $c \in H$ enjoys a continuous representative of the form $f_c:\widehat{\Omega} \to \partial \D$.  One then defines the Schwartzman homomorphism by
\begin{equation}\label{def:schwarz:hom}
\phi(c)
=
\lim_{t \to \infty} \frac{\arg(f_c(T^t \hat \omega)}{t}.
\end{equation}
One of course needs to check that the limit exists and is independent of the chosen representative.  It turns out that the limit in \eqref{def:schwarz:hom} exists for $\widehat{\mu}$ almost every $\widehat{\omega} \in \widehat{\Omega}$, that it is $\widehat{\mu}$-almost surely constant, and that it does not depend on the choice of representative.  The group $ \mathfrak{M} $ is then precisely the image of $H$ under the map $\phi$.

By way of an example, if $\Omega$ is a compact abelian group and $T:\omega \mapsto \omega + \alpha$ is a minimal translation thereof, then the associated $\alpha_n$'s  will be almost-periodic.  It turns out that one can canonically identify $\mathfrak{M}$ with the frequency module of these almost-periodic sequences in this case, so one recovers the standard gap-labeling theorem for almost-periodic operators for free.  See \cite[Example 5.10]{J} for more details on this connection.


\begin{thebibliography}{00}

\bibitem{BGP95} M.\ Baake, U.\ Grimm, C.\ Pisani, Partition function zeros for aperiodic systems, \textit{J. Stat. Phys}\ \textbf{78} (1995), 285--297.

\bibitem{BG01} C.\ A.\ Barata, P.\ S.\ Goldbaum, On the distribution and gap structure of Lee-Yang zeros for the Ising model: Periodic and aperiodic couplings, \textit{J. Stat. Phys.}\ \textbf{103} (2001), 857--891.

\bibitem{BG} J.\ Bochi, N.\ Gourmelon, Some characterizations of domination, \textit{Math.\ Z.}\ \textbf{263} (2009), 221--231.

\bibitem{BHJ03} O.\ Bourget, J.\ Howland, A.\ Joye, Spectral analysis of unitary band matrices, \textit{Commun.\ Math.\ Phys.}\ \textbf{234} (2003), 191--227.

\bibitem{CFKS} H.\ L.\ Cycon, R.\ G.\ Froese, W.\ Kirsch, B.\ Simon, \textit{Schr\"odinger Operators With Applications to Quantum Mechanics and Global Geometry}, Texts and Monographs in Physics, Springer, Berlin, 1987

\bibitem{DL07} D.\ Damanik, D.\ Lenz, Uniform Szeg\H{o} cocycles over strictly ergodic subshifts, \textit{J.\ Approx.\ Theory} \textbf{144} (2007), 133--138.

\bibitem{DMY13} D.\ Damanik, P.\ Munger, W.\ Yessen, Orthogonal polynomials on the unit circle with Fibonacci Verblunsky coefficients, II.~Applications, \textit{J.\ Stat.\ Phys.}\ \textbf{153} (2013), 339--362.

\bibitem{GJ96} J.\ Geronimo, R.\ Johnson, Rotation number associated with difference equations satisfied by polynomials orthogonal on the unit circle, \textit{J.\ Differential Equations} \textbf{132} (1996), 140--178.

\bibitem{GZ06} F.\ Gesztesy, M.\ Zinchenko, Weyl-Titchmarsh theory for CMV operators
associated with orthogonal polynomials on the unit circle, \textit{J.\ Approx.\ Theory} \textbf{139} (2006), 172--213.

\bibitem{GN01} L.\ Golinskii, P.\ Nevai, Szeg\H{o} difference equations, transfer matrices and orthogonal polynomials on the unit circle, \textit{Commun.\ Math.\ Phys.}\ \textbf{223} (2001), 223--259.

\bibitem{J} R.\ Johnson, Exponential dichotomy, rotation number, and linear differential operators with bounded coefficients, \textit{J.\ Diff.\ Eq.}\ \textbf{61} (1986), 54--78.

\bibitem{K08} W.\ Kirsch, An invitation to random Schr\"odinger operators, \textit{Random Schr\"odinger Operators}, Panor.\ Synth\`eses \textbf{25}, 1--119, Soc.\ Math.\ France, Paris, 2008.

\bibitem{KirschMart92} W.\ Kirsch, F.\ Martinelli, On the spectrum of Schr\"odinger operators with a random potential, \textit{Commun.\ Math.\ Phys.}\ \textbf{85} (1982), 329--350.

\bibitem{LS99} Y.\ Last, B.\ Simon, Eigenfunctions, transfer matrices, and absolutely continuous spectrum of one-dimensional Schr\"odinger operators, \textit{Invent.\ Math.}\ \textbf{135} (1999), 329--367.

\bibitem{LY52a} C.\ N.\ Yang, T.\ D.\ Lee, Statistical theory of equations of state and phase transitions. I. Theory of condensation, \textit{Phys. Rev.}\ \textbf{87} (1952), 404--409.

\bibitem{LY52b} C.\ N.\ Yang, T.\ D.\ Lee, Statistical theory of equations of state and phase transitions. II. Lattice gas and Ising model, \textit{Phys. Rev.}\ \textbf{87} (1952), 410--419.

\bibitem{MO} P.\ Munger, D.\ Ong, The H\"older continuity of spectral measures of an extended CMV matrix, preprint (arXiv:1301.0501).

\bibitem{O12} D.\ Ong, Limit-periodic Verblunsky coefficients for orthogonal polynomials on the unit circle, \textit{J.\ Math.\ Anal.\ Appl.}\ \textbf{394} (2012), 633--644.

\bibitem{O14} D.\ Ong, Purely singular continuous spectrum for CMV operators generated by subshifts, \textit{J.\ Stat.\ Phys.}\ \textbf{155} (2014), 763--776.

\bibitem{reedsimon1}  M.\ Reed, B.\ Simon, \textit{Methods of Modern Mathematical Physics, I: Functional Analysis}, Academic Press, New York, 1972.

\bibitem{sacksell1} R.\ Sacker, G.\ Sell, Dichotomies and invariant splittings for linear differential systems I., \textit{J.\ Diff.\ Eq.}\ \textbf{15} (1974), 429--458.

\bibitem{sacksell2} R.\ Sacker, G.\ Sell, A spectral theory for linear differential systems, \textit{J.\ Diff.\ Eq.}\ \textbf{27} (1978), 320--358.

\bibitem{Schwarz} S.\ Schwartzman, Asymptotic cycles, \textit{Ann.\ of Math.}\ \textbf{66} (1957) 270--284.

\bibitem{selgrade} J.\ Selgrade, Isolated invariant sets for flows on vector bundles, \textit{Trans.\ Amer.\ Math.\ Soc.}\ \textbf{203} (1975), 359--390.

\bibitem{S05} B.\ Simon, \textit{Orthogonal Polynomials on the Unit Circle. Part 1. Classical Theory}, American Mathematical Society Colloquium Publications \textbf{54}, Part 1, American Mathematical Society, Providence, RI, 2005.

\bibitem{S05b} B.\ Simon, \textit{Orthogonal Polynomials on the Unit Circle. Part 2. Spectral Theory}, American Mathematical Society Colloquium Publications \textbf{54}, Part 2, American Mathematical Society, Providence, RI, 2005.

\bibitem{T92a} A.\ Teplyaev, Properties of polynomials that are orthogonal on the circle with random parameters, \textit{J.\ Soviet Math.}\ \textbf{61} (1992), 1931--1935.

\bibitem{T92b} A.\ Teplyaev, The pure point spectrum of random orthogonal polynomials on the circle, \textit{Soviet Math.\ Dokl.}\ \textbf{44} (1992), 407--411.

\bibitem{Y04} J.-C.\ Yoccoz, Some questions and remarks about $\mathrm{SL}(2,\R)$ cocycles, \textit{Modern Dynamical Systems and Applications}, 447--458, Cambridge Univ.\ Press, Cambridge, 2004.

\bibitem{Z} Z.\ Zhang, Resolvent set of Schr\"odinger operators and uniform hyperbolicity, preprint (arXiv:1305.4226).

\end{thebibliography}
\end{document}